\documentclass[12pt]{article}
\usepackage{amsfonts}
\usepackage{amssymb,amsthm,latexsym}
\usepackage[numbers,compress]{natbib}
\usepackage{amsmath}
\usepackage{indentfirst}
\allowdisplaybreaks[4]
\newtheorem{theorem}{Theorem}[section]
\newtheorem{prop}[theorem]{Proposition}
\newtheorem{cor}[theorem]{Corollary}
\newtheorem{lemma}[theorem]{Lemma}

\newtheorem{remark}[theorem]{Remark}

\numberwithin{equation}{section}

\begin{document}

\title{ Constant Composition Codes as \\Subcodes of Linear Codes  }

\author{Long Yu{\thanks{Corresponding author.
\newline \indent ~~Email addresses:~longyu@mails.ccnu.edu.cn~(Long Yu),   lxs6682@163.com~(Xiusheng Liu).},~~ Xiusheng Liu}}
\date{  School of Mathematics and Physics, Hubei Polytechnic University, Huangshi,  435003, China}
\maketitle

\begin{abstract}
In this paper, on  one hand, a class of linear codes with one or two weights is obtained. Based on these linear codes, we construct two classes of constant composition codes, which includes optimal constant composition codes depending on LVFC bound.   On the other hand, a class of constant composition codes is derived from  known linear codes.
\end{abstract}



\section{Introduction}
Let $p$ be an odd prime and $q$ be a power of $p$. A linear $[n,k,d]$  code over the finite field $\mathbb{F}_q$ is a $k$-dimensional subspace of $\mathbb{F}_{q}^{n}$ with minimum Hamming distance $d$.  Let $A_i$ denote the number of codewords with Hamming weight $i$ in a linear code $\mathcal{C}$ of length $n$. The weight enumerator of $\mathcal{C}$ is defined by
\[1+A_1X+A_2X^2+\cdots+A_{n}X^{n}.\]
The sequence $(1,A_1,\cdots,A_n)$ is called the weight distribution of the code $\mathcal{C}$.

Let $D=\{d_1,d_2,\cdots,d_n\}\subseteq \mathbb{F}_{p^m}^*$, where $n$, $m$ are  positive integers. Let ${\rm Tr}$ denote the trace function from $\mathbb{F}_{p^m}$ to $\mathbb{F}_p$. We define a linear code of length $n$ over $\mathbb{F}_p$ by
$$\mathcal{C}_D=\{c(a)=({\rm Tr}(ad_1),{\rm Tr}(ad_2),\cdots,{\rm Tr}(ad_n)|a\in \mathbb{F}_{p^m} \}.$$

This construction is generic in the sense that many classes
of known codes could be produced by selecting the defining
set $D\subseteq \mathbb{F}_{p^m}$. So, this  technique could be employed to construct linear codes and cyclic codes in different ways (see \cite{DingNiederreiter2007,DingDing2015,Liyueli2014,Liyue2014,LuoandFeng1,LuoandFeng2,yuliuamc,yu-liu2016,
Zeng2012,ZengHujiang2010,zhengwanghzeng2015,zhengwangyuliu2014,zhouding2014,ZhouLiFanH2016}, and references theirin).

Let $S=\{s_0,\cdots,s_{q-1}\}$ be an alphabet of size $q$.  An $[n,M,d,(\omega_0,\omega_1,\cdots,\omega_{q-1})_q]$ constant composition code(CCC) is a subset $C\subset S^n$ of size $M$, minimal distance $d$ and where the element $s_i$ occurs exactly $\omega_i$ times in each codeword in $C$.

Constant composition codes were studied
already in the 1960s. Both algebraic and combinatorial constructions of CCCs
have been proposed. For further information, the reader is referred to \cite{ChuColbournDukes2004,Ding2008,Dingyin2005,Dingyuan2005,LuoHellseth2011}.

The objective of this paper is to construct two classes of CCCs. 
A new constructions CCCs which are subcodes of linear
codes (not cyclic codes) is proposed. On one hand, 
we   define a class of linear codes $\mathcal{C}_{D(\alpha)}$ by a set $D(\alpha)$. The corresponding exponential sums have close relationship with Gauss sums. By using the technology of finite field, the parameters of linear codes $\mathcal{C}_{D(\alpha)}$ are obtained for all $\alpha\in \mathbb{F}_p$ (see Theorem~\ref{th:diyileizhongliang}). Furthermore, we select a kind of set $S$ and construct a class of CCCs whose  parameters is presented (see Theorem~\ref{th:yileiccc}). Some of these CCCs are optimal in the sense that they meet LFVC bound. On the other hand, a class of CCCs is constructed from known linear codes (see Theorem~\ref{th:erleiccc}).


\section{Preliminaries}
Throughout this paper, we let $q=p^m$, where $m$ is a positive integer. Let $\eta$ and $\overline{\eta}$ be the quadratic multiplicative character on $\mathbb{F}_q$ and $\mathbb{F}_p$, respectively. Let $\chi_1(\cdot)=\zeta_p^{{\rm Tr}(\cdot)}$ and $\overline{\chi}_1=\zeta_p^{(\cdot)}$ be the canonical additive characters on $\mathbb{F}_q$ and $\mathbb{F}_p$, respectively. We define $\eta(0)=0=\overline{\eta}(0)$, then the quadratic Gaussian sum $G(\eta,\chi_1)$ on  $\mathbb{F}_q$ is defined by
\[G(\eta,\chi_1)=\sum_{x\in \mathbb{F}_q}\eta(x)\chi_1(x) ,\]
and  the quadratic Gaussian sum $G(\overline{\eta},\overline{\chi}_1)$ on  $\mathbb{F}_p$ is defined by
\[G(\overline{\eta},\overline{\chi}_1)=\sum_{x\in \mathbb{F}_p}\overline{\eta}(x)\overline{\chi}_1(x) .\]

The following  results are well known.
\begin{lemma}\cite{Lidl R}
Let the notations be given as above, we have
$$G(\eta,\chi_1)=(-1)^{m-1}\sqrt{-1}^{(\frac{p-1}{2})^2m} \sqrt{q}$$
and
$$G(\overline{\eta},\overline{\chi}_1)=\sqrt{-1}^{(\frac{p-1}{2})^2} \sqrt{p}.$$
\end{lemma}
 \begin{lemma}\cite{Lidl R}\label{lem:ercihanshuqiuhe}
Let $\chi$ be a nontrivial additive character of $\mathbb{F}_q$, and let $f(x)=a_2x^2+a_1x+a_0\in \mathbb{F}_q[x]$ with $a_2\neq0$. Then
\[\sum_{x\in \mathbb{F}_q}\chi(f(x))=\chi(a_0-a_1^2/(4a_2))\eta(a_2)G(\eta,\chi).\]
\end{lemma}

The conclusions of the following two lemmas are easy to obtain.
\begin{lemma}\label{lem:ercitezheng}
If $m$ is odd, then $\eta(a)=\overline{\eta}(a)$ for any $a\in \mathbb{F}_p$. If $m$ is even,  then $\eta(a)=1$ for any $a\in \mathbb{F}_p^*$.
\end{lemma}
\begin{lemma}\label{lem:diyileichangdu}
For each  $\alpha\in \mathbb{F}_p$, we let $$N_\alpha=\#\{x\in \mathbb{F}_{p^m}| {\rm Tr}(x)=\alpha\}.$$  Then
$N_\alpha =   p^{m-1}.$
\end{lemma}

We will need the following lemma.
\begin{lemma}\cite{DingDing2015}\label{lem:changdu}
With the notations given as above. For each $\alpha\in \mathbb{F}_p$, let $$N_\alpha=\#\{x\in \mathbb{F}_{p^m}| {\rm Tr}(x^2)=\alpha\}.$$ Then
$$
N_\alpha=\left\{
  \begin{array}{ll}
    p^{m-1}, & \hbox{if $m$ is odd and $\alpha=0$;} \\
    p^{m-1}-(-1)^{(\frac{p-1}{2})^2\frac{m}{2}}(p-1)p^{\frac{m-2}{2}}, & \hbox{if $m$ is even and $\alpha=0$;} \\
     p^{m-1}+\overline{\eta}(-\alpha)(-1)^{(\frac{p-1}{2})^2(\frac{m+1}{2})}p^{\frac{m-1}{2}}, & \hbox{if $m$ is odd and $\alpha\neq0$;} \\
    p^{m-1}+(-1)^{(\frac{p-1}{2})^2\frac{m}{2}}p^{\frac{m-2}{2}}, & \hbox{if $m$ is even and $\alpha\neq0$.}
  \end{array}
\right.
$$
\end{lemma}

%

At the end of this section, we give the LFVC bound of constant composition code.
\begin{prop}\cite{LFVCbound2003}
Assume $nd-n^2+\omega_0^2+\omega_1^2+\cdots+\omega_{p-1}^2>0$. Then, an $[n,M,d,(\omega_\beta)_{\beta\in\mathbb{F}_p}]$ CCC satisfies the following inequality
\[M\leq nd/\left(nd-n^2+\omega_0^2+\omega_1^2+\cdots+\omega_{p-1}^2\right).\]
If $$M=nd/\left(nd-n^2+\omega_0^2+\omega_1^2+\cdots+\omega_{p-1}^2\right),$$ then we call CCC is optimal.
\end{prop}

\section{The first construction}
In this section, we will define  a CCC as a  subcode of   linear code $\mathcal{C}_{D(\alpha)}$ defined by (\ref{eq:1}). Throughout this section, we always assume that $m$ is a positive integer. The defining set $D(\alpha)$ is given by
\[D(\alpha)=\{d\in \mathbb{F}_{p^m}^*| {\rm Tr}(d)=\alpha\},\]
where $\alpha\in \mathbb{F}_p$. Let $n_\alpha$ be the length of linear code $\mathcal{C}_{D(\alpha)}$, where
\begin{equation}\label{eq:1}
   \mathcal{C}_{D(\alpha)}=\{c(a)=({\rm Tr}(ad_1),{\rm Tr}(ad_2),\cdots,{\rm Tr}(ad_{n_\alpha})|a\in \mathbb{F}_{p^m} \}.
\end{equation}
Then by Lemma~\ref{lem:diyileichangdu}, we have the following result.
\begin{lemma}\label{lem:changdu3.1}
With the notations given as above, we have
$$n_\alpha=\left\{
  \begin{array}{ll}
    p^{m-1}-1, & \hbox{$\alpha=0$;} \\
    p^{m-1}, & \hbox{otherwise.}
  \end{array}
\right.$$
\end{lemma}
%

\begin{theorem}\label{th:diyileizhongliang}
Let  the notations be given as above.
\begin{itemize}
  \item When $\alpha=0$,  $\mathcal{C}_{D(0)}$ defined by (\ref{eq:1}) is a $[p^{m-1}-1,m-1]$ code over $\mathbb{F}_p $ with the weight distribution as follows.
\begin{center}
\begin{tabular}{|c|c|}
  \hline
  weight & frequency  \\  \hline
  $0$ & $1$   \\  \hline
  $p^{m-2}(p-1)$ & $p^{m-1}-1$   \\  \hline
 \end{tabular}
\end{center}
  \item When $\alpha\in \mathbb{F}_p^*$,  $\mathcal{C}_{D(\alpha)}$ defined by (\ref{eq:1}) is a $[p^{m-1},m]$ code over $\mathbb{F}_p $ with the weight distribution as follows.
\begin{center}
\begin{tabular}{|c|c|}
  \hline
  weight & frequency  \\  \hline
$0$ & $1$  \\  \hline
  $p^{m-1}$ & $p-1$   \\  \hline
  $p^{m-2}(p-1)$ & $p^m-p$   \\  \hline
 \end{tabular}
\end{center}
\end{itemize}

\end{theorem}

\begin{proof}
For $a\neq0$, the Hamming weight of codeword $c(a)$ is equal to
\begin{eqnarray}\label{eq:zhongliang}
\nonumber wt(c(a))  &=& n_\alpha-\frac{1}{p}\sum_{x\in {D(\alpha)}}\sum_{u\in \mathbb{F}_p}\zeta_p^{u{\rm Tr}(ax)}  \\
\nonumber   &=&  n_\alpha-\frac{1}{p^2}\sum_{x\in \mathbb{F}_{p^m}^* }\sum_{u\in \mathbb{F}_p}\zeta_p^{u{\rm Tr}(ax)}\sum_{v\in \mathbb{F}_p}\zeta_p^{v({\rm Tr}(x)-\alpha)} \\
\nonumber &=&   n_\alpha-\frac{1}{p^2}\sum_{x\in \mathbb{F}_{p^m}^* }(1+\sum_{u\in \mathbb{F}_p^*}\zeta_p^{u{\rm Tr}(ax)})(1+\sum_{v\in \mathbb{F}_p^*}\zeta_p^{v({\rm Tr}(x)-\alpha)}) \\
\nonumber   &=&  n_\alpha-\frac{1}{p^2}(p^m-1)-\frac{1}{p^2}\sum_{x\in \mathbb{F}_{p^m}^* }\sum_{v\in \mathbb{F}_p^*}\zeta_p^{v({\rm Tr}(x)-\alpha)}-\frac{1}{p^2}\sum_{x\in \mathbb{F}_{p^m}^* }\sum_{u\in \mathbb{F}_p^*}\zeta_p^{u{\rm Tr}(ax)} \\
\nonumber& &-\frac{1}{p^2}\sum_{x\in \mathbb{F}_{p^m}^* }\sum_{u\in \mathbb{F}_p^*}\zeta_p^{u{\rm Tr}(ax)}\sum_{v\in \mathbb{F}_p^*}\zeta_p^{v({\rm Tr}(x)-\alpha)}\\
\nonumber   &=&   n_\alpha-\frac{1}{p^2}(p^m-1)+\frac{1}{p^2}\sum_{v\in \mathbb{F}_p^*}\zeta_p^{-v\alpha}+\frac{1}{p^2}(p-1)\\
&&-\frac{1}{p^2}\sum_{u\in \mathbb{F}_p^*}\sum_{v\in \mathbb{F}_p^*}\zeta_p^{-v\alpha}(\sum_{x\in \mathbb{F}_{p^m} }\zeta_p^{{\rm Tr}((au+v)x)}-1).
\end{eqnarray}

1) If $\alpha=0$, from (\ref{eq:zhongliang}) and Lemma~\ref{lem:changdu3.1}, we have
\begin{eqnarray*}
  wt(c(a))    &=&   n_0-\frac{1}{p^2}(p^m-1)+\frac{2}{p^2}(p-1)-\frac{1}{p^2}\sum_{u\in \mathbb{F}_p^*}\sum_{v\in \mathbb{F}_p^*} (\sum_{x\in \mathbb{F}_{p^m} }\zeta_p^{{\rm Tr}((au+v)x)}-1)\\
 &=&  n_0-p^{m-2}+1-\frac{1}{p^2}\sum_{u\in \mathbb{F}_p^*}\sum_{v\in \mathbb{F}_p^*} \sum_{x\in \mathbb{F}_{p^m} }\zeta_p^{{\rm Tr}((au+v)x)}\\
&=&  n_0-p^{m-2}+1-p^{m-2}\#\{u,v\in \mathbb{F}_p^*|au+v=0 \}\\
&=& \left\{
      \begin{array}{ll}
       0 , & \hbox{if $a\in \mathbb{F}_p^*$;} \\
        p^{m-2}(p-1), & \hbox{otherwise.}
      \end{array}
    \right.
\end{eqnarray*}
Note that when $a\in \mathbb{F}_p$, we have $wt(c(a))=0$. This implies that the dimension of linear code $\mathcal{C}_{D(0)}$ is $m-1$.

2) If $\alpha\neq0$, then
\[\sum_{v\in \mathbb{F}_p^*}\zeta_p^{-v\alpha}=-1 ~~{\rm and}~~ \sum_{u\in \mathbb{F}_p^*}\sum_{v\in \mathbb{F}_p^*}\zeta_p^{-v\alpha}=-(p-1).\]
From (\ref{eq:zhongliang}) and Lemma~\ref{lem:changdu3.1}, we get
\begin{eqnarray*}
  wt(c(a))
 &=& n_\alpha-p^{m-2}-\frac{1}{p^2}\sum_{u\in \mathbb{F}_p^*}\sum_{v\in \mathbb{F}_p^*} \zeta_p^{-v\alpha}\sum_{x\in \mathbb{F}_{p^m} }\zeta_p^{{\rm Tr}((au+v)x)}\\
&=& \left\{
      \begin{array}{ll}
       n_\alpha-p^{m-2}-p^{m-2}\sum\limits_{\{u,v\in \mathbb{F}_p^*|au+v=0\}}\zeta_p^{-v\alpha} , & \hbox{if $a\in \mathbb{F}_p^*$;} \\
        p^{m-2}(p-1), & \hbox{otherwise.}
      \end{array}
    \right.\\
&=& \left\{
      \begin{array}{ll}
       n_\alpha-p^{m-2}-p^{m-2}\sum\limits_{v\in \mathbb{F}_p^*}\zeta_p^{-v\alpha} , & \hbox{if $a\in \mathbb{F}_p^*$;} \\
        p^{m-2}(p-1), & \hbox{otherwise.}
      \end{array}
    \right.\\
&=& \left\{
      \begin{array}{ll}
       p^{m-1} , & \hbox{if $a\in \mathbb{F}_p^*$;} \\
        p^{m-2}(p-1), & \hbox{otherwise.}
      \end{array}
    \right.
\end{eqnarray*}
Note that when $a\in \mathbb{F}_{p^m}^*$, we have $wt(c(a))>0$. This implies that the dimension of linear code $\mathcal{C}_{D(\alpha)}$ is $m$.

\end{proof}

The Code $\mathcal{C}_{\alpha}'$ is defied by
\begin{equation}\label{eq:2}
 \mathcal{C}_{\alpha}'=\{c(a)~|~a\in \mathbb{F}_{p^m}\setminus \mathbb{F}_{p}\}.
 \end{equation}
Denote the size of code $\mathcal{C}_{\alpha}'$ by $M_\alpha.$
It is easy to check
\begin{eqnarray*}
  M_\alpha &=& \left\{
                 \begin{array}{ll}
                   p^{m-1}-1, & \hbox{if $\alpha=0$;} \\
                   p^{m}-p, & \hbox{otherwise.}
                 \end{array}
               \right.
\end{eqnarray*}
%
\begin{theorem}\label{th:yileiccc}
The code $\mathcal{C}_{\alpha}'$ defined by (\ref{eq:2}) is an  $[n_\alpha,M_\alpha,d,(\omega_\beta)_{\beta\in \mathbb{F}_{p}}]$  CCC, where

1)  in the case  $\alpha=0$:
\begin{eqnarray*}
  n_0 &=& p^{m-1}-1; \\
M_0 &=& p^{m-1}-1;\\
  \omega_0 &=& p^{m-2}-1; \\
  \omega_\beta &=&  p^{m-2}~~~~\mbox{ for any} ~~~\beta\in \mathbb{F}_{p}^*;\\
  d &=&   p^{m-2}(p-1);
\end{eqnarray*}

2)  in the case  $\alpha\neq0$:
\begin{eqnarray*}
  n_\alpha &=& p^{m-1}; \\
M_\alpha &=& p^m-p;\\
    \omega_\beta &=&  p^{m-2}~~~~\mbox{ for any} ~~~\beta\in \mathbb{F}_{p};\\
  d &=&   p^{m-2}(p-1).
\end{eqnarray*}
\end{theorem}
\begin{proof}
1) Note that $a\neq0$. If $\alpha=0$, for any $\beta\in \mathbb{F}_{p}^*$, we have
\begin{eqnarray*}
  \omega_\beta &=& \frac{1}{p}\sum_{x\in {D(\alpha)}}\sum_{u\in \mathbb{F}_{p}}\zeta_p^{u({\rm Tr}(ax)-\beta)} \\
    &=& \frac{1}{p^2}\sum_{x\in \mathbb{F}_{p^m}^*}\sum_{u\in \mathbb{F}_{p}}\zeta_p^{u({\rm Tr}(ax)-\beta)}\sum_{v\in \mathbb{F}_{p}}\zeta_p^{v{\rm Tr}(x)} \\
&=& \frac{1}{p^2}\sum_{x\in \mathbb{F}_{p^m}}\sum_{u\in \mathbb{F}_{p}}\zeta_p^{u({\rm Tr}(ax)-\beta)}\sum_{v\in \mathbb{F}_{p}}\zeta_p^{v{\rm Tr}(x)} \\
    &=& \frac{1}{p^2}\sum_{u\in \mathbb{F}_{p}}\sum_{v\in \mathbb{F}_{p}}\zeta_p^{-u\beta}\sum_{x\in \mathbb{F}_{p^m}}\zeta_p^{{\rm Tr}((au+v)x)}\\
&=&p^{m-2}+ \frac{1}{p^2}\sum_{u\in \mathbb{F}_{p}^*}\zeta_p^{-u\beta}\sum_{x\in \mathbb{F}_{p^m}}\zeta_p^{{\rm Tr}(aux)}+\frac{1}{p^2}\sum_{v\in \mathbb{F}_{p}^*}\sum_{x\in \mathbb{F}_{p^m}}\zeta_p^{{\rm Tr}(vx)}\\
&& +\frac{1}{p^2}\sum_{u\in \mathbb{F}_{p}^*}\sum_{v\in \mathbb{F}_{p}^*}\zeta_p^{-u\beta}\sum_{x\in \mathbb{F}_{p^m}}\zeta_p^{{\rm Tr}((au+v)x)}\\
&=& p^{m-2}+\frac{1}{p^2}\sum_{u\in \mathbb{F}_{p}^*}\sum_{v\in \mathbb{F}_{p}^*}\zeta_p^{-u\beta}\sum_{x\in \mathbb{F}_{p^m}}\zeta_p^{{\rm Tr}((au+v)x)}.
\end{eqnarray*}
Note that $a\in \mathbb{F}_{p^m}\setminus \mathbb{F}_{p}$, then $au+v\neq 0$ for any $u,v\in \mathbb{F}_{p}^*$. This implies that
$$\sum_{x\in \mathbb{F}_{p^m}}\zeta_p^{{\rm Tr}((au+v)x)}=0.$$
Thus, we have $\omega_\beta=p^{m-2}$, for any $\beta\neq0$. Recall $n_0=p^{m-1}-1 $, then $\omega_0= p^{m-2}-1.$


2) If $\alpha\neq0$, for any $\beta\in \mathbb{F}_{p}^*$ and $a\neq0$, then
\begin{eqnarray*}
  \omega_\beta &=& \frac{1}{p}\sum_{x\in {D(\alpha)}}\sum_{u\in \mathbb{F}_{p}}\zeta_p^{u({\rm Tr}(ax)-\beta)} \\
    &=& \frac{1}{p^2}\sum_{x\in \mathbb{F}_{p^m}^*}\sum_{u\in \mathbb{F}_{p}}\zeta_p^{u({\rm Tr}(ax)-\beta)}\sum_{v\in \mathbb{F}_{p}}\zeta_p^{v({\rm Tr}(x)-\alpha)} \\
 &=& \frac{1}{p^2}\sum_{x\in \mathbb{F}_{p^m}}\sum_{u\in \mathbb{F}_{p}}\zeta_p^{u({\rm Tr}(ax)-\beta)}\sum_{v\in \mathbb{F}_{p}}\zeta_p^{v({\rm Tr}(x)-\alpha)} \\
    &=& \frac{1}{p^2}\sum_{u\in \mathbb{F}_{p}}\sum_{v\in \mathbb{F}_{p}}\zeta_p^{-u\beta}\zeta_p^{-v\alpha}\sum_{x\in \mathbb{F}_{p^m}}\zeta_p^{{\rm Tr}((au+v)x)}\\
&=&p^{m-2}+ \frac{1}{p^2}\sum_{u\in \mathbb{F}_{p}^*}\zeta_p^{-u\beta}\sum_{x\in \mathbb{F}_{p^m}}\zeta_p^{{\rm Tr}(aux)}+\frac{1}{p^2}\sum_{v\in \mathbb{F}_{p}^*}\zeta_p^{-v\alpha}\sum_{x\in \mathbb{F}_{p^m}}\zeta_p^{{\rm Tr}(vx)}\\
&& +\frac{1}{p^2}\sum_{u\in \mathbb{F}_{p}^*}\sum_{v\in \mathbb{F}_{p}^*}\zeta_p^{-u\beta}\zeta_p^{-v\alpha}\sum_{x\in \mathbb{F}_{p^m}}\zeta_p^{{\rm Tr}((au+v)x)}\\
&=& p^{m-2}+\frac{1}{p^2}\sum_{u\in \mathbb{F}_{p}^*}\sum_{v\in \mathbb{F}_{p}^*}\zeta_p^{-u\beta}\zeta_p^{-v\alpha}\sum_{x\in \mathbb{F}_{p^m}}\zeta_p^{{\rm Tr}((au+v)x)}.
\end{eqnarray*}
Note that $a\in \mathbb{F}_{p^m}\setminus \mathbb{F}_{p}$, then $au+v\neq 0$ for any $u,v\in \mathbb{F}_{p}^*$.
This leads to
$$\sum_{x\in \mathbb{F}_{p^m}}\zeta_p^{{\rm Tr}((au+v)x)}=0, ~~u,v\in  \mathbb{F}_{p}^*.$$
Therefore, $\omega_\beta=p^{m-2}$ for any $\beta\in \mathbb{F}_{p}^*$. Recall $n_\alpha=p^{m-1}$, then $\omega_0=p^{m-1}-p^{m-2}(p-1)=p^{m-2}$.

Denote by $d_H(c(a_1),c(a_2))$ the Hamming distance of $c(a_1)$ and $(a_2)$. For any $\alpha\in \mathbb{F}_p$,  when $a_1$ and $a_2$ run through $\mathbb{F}_{p^m}\setminus \mathbb{F}_{p}$ with $a_1\neq a_2$, then $a_1- a_2$ runs through $\mathbb{F}_{p^m}^*$. Therefore, the minimal distance of $\mathcal{C}'$ is the same as that of $\mathcal{C}$, which is obtained by Theorem~\ref{th:diyileizhongliang}. This completes the proof.
\end{proof}
\begin{remark}
When $\alpha=0$, the parameters of code $\mathcal{C}_{0}'$ in above Theorem satisfy $$n_0d/(n_0d-n_0^2+\omega_0^2+\omega_1^2+\cdots+\omega_{p-1}^2)=M_0.$$
Thus, we obtain a class of optimal CCCs. However, if $\alpha\in \mathbb{F}_p^*$, the parameters of code $\mathcal{C}_{\alpha}'$    satisfy $$n_\alpha d-n_\alpha^2+\omega_0^2+\omega_1^2+\cdots+\omega_{p-1}^2=0.$$
Therefore, the LFVC bound cannot be applied to measure the optimality of these CCCs. Note that for large $m$, code $\mathcal{C}_{\alpha}'$ has the same minimal distance $d$ and length $n_\alpha$. In this sense, compared to $\mathcal{C}_0'$ which is optimal, $\mathcal{C}_{\alpha}'$ has more codewords with same minimal distance $d$  for $\alpha\in \mathbb{F}_p^*$. This implies that $\mathcal{C}_{\alpha}'$ are excellent.


\end{remark}
%

\section{The second construction}
In this section, we  let $\tau=(-1)^{(\frac{p-1}{2})^2\frac{m}{2}}$, where $m$ is even. The defining set $E$ is given by
\[{E}=\{d\in \mathbb{F}_{p^m}^*| {\rm Tr}(d^2)=0\}.\]
The linear code $\mathcal{C}_{E}$ is defined as
 \begin{equation}\label{eq:3}
 \mathcal{C}_{E}=\{c(a)=({\rm Tr}(ad_1),{\rm Tr}(ad_2),\cdots,{\rm Tr}(ad_n)|a\in  \mathbb{F}_{p^m} \}.
 \end{equation}
\begin{lemma}\cite{DingDing2015}\label{lem:dierleizongliang}
Then the code $\mathcal{C}_{E}$ defined by (\ref{eq:3}) over $\mathbb{F}_p$ has parameters $[ p^{m-1}-\tau(p-1)p^{\frac{m}{2}-1}-1,m]$ and weight distribution in the following.
\begin{center}
\begin{tabular}{|c|c|}
  \hline
 weight & frequency  \\ \hline
 $0$ & $1$  \\ \hline
  $(p-1)p^{m-2}$ & $p^{m-1}-\tau(p-1)p^{\frac{m}{2}-1}-1$ \\ \hline
  $(p-1)(p^{m-2}-\tau p^{\frac{m}{2}-1}) $& $ (p-1) \left( p^{m-1}+\tau p^{\frac{m}{2}-1}\right)$  \\  \hline
\end{tabular}
\end{center}\end{lemma}

Define set
\[S=\{a\in \mathbb{F}_{q^m}| {\rm Tr}(a^2)\neq0\}\]
and code 
\begin{equation}\label{eq:4}
\mathcal{C}_E'=\{c(a)~|~a\in S\}.
\end{equation}
By Lemma~\ref{lem:changdu}, we obtain the size of $S$ is $p^m-p^{m-1}+\tau(p-1)p^{\frac{m}{2}-1}$.
In the following, we give our main result.
\begin{theorem}\label{th:erleiccc}
The code $\mathcal{C}_E'$  defined by (\ref{eq:4})  is a CCC with parameters $[n,M,d,(\omega_\beta)_{\beta\in \mathbb{F}_p}]$ where
\begin{eqnarray*}
  n &=&  p^{m-1}-\tau(p-1)p^{\frac{m}{2}-1}-1; \\
  M &=& p^m-p^{m-1}+\tau(p-1)p^{\frac{m}{2}-1}; \\
  \omega_0 &=& p^{m-2}-1; \\
 \omega_\beta &=& p^{m-2}-\tau p^{\frac{m}{2}-1} ~~for~~ any~~\beta\in \mathbb{F}_p^*;  \\
  d &=& \left\{
          \begin{array}{ll}
            (p-1)p^{m-2}, & \hbox{$\tau=-1$;} \\
            (p-1)(p^{m-2}- p^{\frac{m}{2}-1}), & \hbox{$\tau=1$.}
          \end{array}
        \right.
\end{eqnarray*}
\end{theorem}

\begin{proof}
%
It is easy to see $$n=p^{m-1}-\tau(p-1)p^{\frac{m}{2}-1}-1$$ and $$ M=p^m-p^{m-1}+\tau(p-1)p^{\frac{m}{2}-1}.$$
For any $\beta\in \mathbb{F}_{p}^*$,  we have
\begin{eqnarray}\label{eq:dierleiwfenbu}
\nonumber  \omega_\beta &=& \frac{1}{p}\sum_{x\in {E}}\sum_{u\in \mathbb{F}_{p}}\zeta_p^{u({\rm Tr}(ax)-\beta)} \\
\nonumber    &=& \frac{1}{p^2}\sum_{x\in \mathbb{F}_{p^m}^*}\sum_{u\in \mathbb{F}_{p}}\zeta_p^{u({\rm Tr}(ax)-\beta)}\sum_{v\in \mathbb{F}_{p}}\zeta_p^{v{\rm Tr}(x^2)} \\
\nonumber&=& \frac{1}{p^2}\sum_{x\in \mathbb{F}_{p^m}}\sum_{u\in \mathbb{F}_{p}}\zeta_p^{u({\rm Tr}(ax)-\beta)}\sum_{v\in \mathbb{F}_{p}}\zeta_p^{v{\rm Tr}(x^2)} \\
\nonumber    &=& \frac{1}{p^2}\sum_{u\in \mathbb{F}_{p}}\sum_{v\in \mathbb{F}_{p}}\zeta_p^{-u\beta}\sum_{x\in \mathbb{F}_{p^m}}\zeta_p^{{\rm Tr}(aux+vx^2)}\\
\nonumber&=& p^{m-2}+\frac{1}{p^2}\sum_{u\in \mathbb{F}_{p}^*}\zeta_p^{-u\beta}\sum_{x\in \mathbb{F}_{p^m}}\zeta_p^{{\rm Tr}(aux)}+\frac{1}{p^2}\sum_{v\in \mathbb{F}_{p}^*}\sum_{x\in \mathbb{F}_{p^m}}\zeta_p^{{\rm Tr}(vx^2)}\\
& &+\frac{1}{p^2}\sum_{u\in \mathbb{F}_{p}^*}\sum_{v\in \mathbb{F}_{p}^*}\zeta_p^{-u\beta} \sum_{x\in \mathbb{F}_{p^m}}\zeta_p^{{\rm Tr}(aux+vx^2)}.
\end{eqnarray}
Since $a\cdot u\neq0$, then $$\sum_{x\in \mathbb{F}_{p^m}}\zeta_p^{{\rm Tr}(aux)}=0.$$
By Lemmas~\ref{lem:ercihanshuqiuhe} and \ref{lem:ercitezheng}, together with $a\in S$,  we have (\ref{eq:dierleiwfenbu}) is equal to
\begin{eqnarray*}
  \omega_\beta
&=& p^{m-2}+\frac{p-1}{p^2}G(\eta,\chi_1)+\frac{1}{p^2}\sum_{u\in \mathbb{F}_{p}^*}\sum_{v\in \mathbb{F}_{p}^*}\zeta_p^{-u\beta}  \zeta_p^{{\rm Tr}( -\frac{a^2u^2}{4v})}G(\eta,\chi_1)\\
&=& p^{m-2}+\frac{p-1}{p^2}G(\eta,\chi_1)+\frac{1}{p^2}G(\eta,\chi_1)\sum_{v\in \mathbb{F}_{p}^*} (\sum_{u\in \mathbb{F}_{p}}\zeta_p^{-\beta u- \frac{{\rm Tr}(a^2)}{4v}u^2}-1)\\
&=& p^{m-2}+\frac{1}{p^2}G(\eta,\chi_1)\sum_{v\in \mathbb{F}_{p}^*} \sum_{u\in \mathbb{F}_{p}}\zeta_p^{-\beta u- \frac{{\rm Tr}(a^2)}{4v}u^2}\\
&=& p^{m-2}+\frac{1}{p^2}G(\eta,\chi_1)\sum_{v\in \mathbb{F}_{p}^*}\overline{\eta}(-\frac{{\rm Tr}(a^2)}{4v})\zeta_p^{\frac{\beta^2v}{{\rm Tr}(a^2)}}G(\overline{\eta},\overline{\chi}_1)\\
&=& p^{m-2}+\frac{1}{p^2}G(\eta,\chi_1)G(\overline{\eta},\overline{\chi}_1)\sum_{v\in \mathbb{F}_{p}^*}\overline{\eta}(-1)\overline{\eta}(\frac{\beta^2v}{{\rm Tr}(a^2)})\zeta_p^{\frac{\beta^2v}{{\rm Tr}(a^2)}}\\
&=& p^{m-2}+\frac{1}{p^2}G(\eta,\chi_1)G(\overline{\eta},\overline{\chi}_1)\overline{\eta}(-1)\sum_{v\in \mathbb{F}_{p}^*}\overline{\eta}(v)\zeta_p^{v}\\
&=& p^{m-2}+\frac{1}{p^2}G(\eta,\chi_1)G^2(\overline{\eta},\overline{\chi}_1)\overline{\eta}(-1)\\
&=& p^{m-2}-(-1)^{(\frac{p-1}{2})^2} \overline{\eta}(-1)\tau p^{\frac{m}{2}-1}\\
&=& p^{m-2}-(-1)^{(\frac{p-1}{2})^2}(-1)^{(\frac{p-1}{2})}\tau p^{\frac{m}{2}-1}\\
&=& p^{m-2}-\tau p^{\frac{m}{2}-1}.
\end{eqnarray*}
Note that $n = p^{m-1}-\tau(p-1)p^{\frac{m}{2}-1}-1$, then
\[\omega_0=p^{m-2}-1.\]

Denote by $d_H(c(a_1),c(a_2))$ the Hamming distance of $c(a_1)$ and $(a_2)$. When $a_1$ and $a_2$ run through $S$ with $a_1\neq a_2$, then $a_1- a_2$ runs through $\mathbb{F}_{p^m}^*$. Therefore, the minimal distance of $\mathcal{C}_E'$ is the same as that of $\mathcal{C}_{E}$. By Lemma~\ref{lem:dierleizongliang}, we finish the proof.
\end{proof}
\begin{remark}
We can check that $nd-n^2+\omega_0^2+\omega_1^2+\cdots+\omega_{p-1}^2<0 $ for $\tau=\pm 1.$ Therefore, the LFVC bound cannot be applied to measure the optimality of these CCCs.
\end{remark}

In a similar way, we can  prove that $\{c(a)|a\in \mathbb{F}_{p^m}^*\setminus S \}$ is a CCC. However, this code has only $p^{m-1}-\tau (p-1)p^{\frac{m-2}{2}}-1$ codewords which is  fewer in comparison with $\mathcal{C}_E'$.
\begin{cor}\label{th:erlei}
The code  $\{c(a)|a\in \mathbb{F}_{p^m}^*\setminus S \}$     is a CCC with parameters $[n,M',d,(\omega_\beta)_{\beta\in \mathbb{F}_p}]$ where
\begin{eqnarray*}
  n &=&  p^{m-1}-\tau(p-1)p^{\frac{m}{2}-1}-1; \\
  M' &=& p^{m-1}-\tau(p-1)p^{\frac{m}{2}-1}-1; \\
  \omega_0 &=& p^{m-2}-\tau(p-1)p^{\frac{m}{2}-1}-1; \\
 \omega_\beta &=& p^{m-2} ~~for~~ any~~\beta\in \mathbb{F}_p^*;  \\
  d &=& \left\{
          \begin{array}{ll}
            (p-1)p^{m-2}, & \hbox{$\tau=-1$;} \\
            (p-1)(p^{m-2}- p^{\frac{m}{2}-1}), & \hbox{$\tau=1$.}
          \end{array}
        \right.
\end{eqnarray*}
\end{cor}

\section{Conclusion}
In this paper, we obtained several classes of CCCs with exact parameters.
One of them is optimal based on LFVC bound. However, the other  CCCs cannot be measured by LFVC bound.
The optimality of these   CCCs is not clear. At the end, we mention that CCCs over $\mathbb{F}_p$ obtained in this paper can be generalized to the case of  CCCs over $\mathbb{F}_q$.

\bigbreak

\bigbreak
\textbf{Acknowledgment}

The authors would like to thank the anonymous referees for their comments that improved the
presentation of this paper. The work of L. Yu was support by  research
funds of HBPU(Grant No. 17xjz04R).




\begin{thebibliography}{99}
\bibitem{ChuColbournDukes2004} W. Chu, C.J.Colbourn and  P. Dukes, "Constructions for permutation
codes in powerline communnications," Designs Codes, Crypt., \textbf{32}, 51--64, 2004.


\bibitem{Ding2008}C. Ding, "Optimal constant composition codes from zero-difference
balanced functions," \emph{IEEE Trans. Inf. Theory},  \textbf{54}(12), 5766--5770,   2008.

\bibitem{Dingyin2005}C. Ding and J. Yin, "Algebraic constructions of constant composition
codes," \emph{IEEE Trans. Inf. Theory},  \textbf{51}(4),   1585--1589, 2005.



%
\bibitem{Dingyuan2005}C. Ding and J. Yuan, "A family of optimal constant-composition
codes," \emph{IEEE Trans. Inf. Theory},   \textbf{51}(10),   3668--3671, 2005.
\bibitem{DingNiederreiter2007} C. Ding and H. Niederreiter, "Cyclotomic linear codes of order $3$", \emph{IEEE Trans. Inf. Theory}, \textbf{53}(6), 2274--2277, 2007.
\bibitem{DingDing2015} K. Ding   and C. Ding. "A Class of Two-Weight and Three-Weight Codes and Their Applications in Secret Sharing." \emph{IEEE Trans. Inf. Theory}, \textbf{61}(11), 5835--5842, 2015.


\bibitem{Lidl R} R. Lidl and H. Niederreiter, ``Finite Fields,'' Cambridge, U.K.: Cambridge Univ. Press, 1997.








\bibitem{Liyueli2014} C. Li, Q. Yue and F. Li, "Hamming weights of the duals of cyclic codes with two zeros", \emph{IEEE Trans. Inform. Theory}, \textbf{60}(7), 3895--3902, 2014.


\bibitem{Liyue2014}  C. Li and Q. Yue, "Weight distributions of two classes of cyclic codes with respect to two distinct order elements", \emph{IEEE Trans. Inform. Theory}, \textbf{60}(1), 296--303, 2014.
\bibitem{LFVCbound2003}Y. Luo, F.  Fu, A.  Vinck and W.  Chen, "On constantcomposition codes over $\mathbb{Z}_q$," \emph{IEEE Trans. Inform. Theory}, \textbf{49}, 3010--3016, 2003.

\bibitem{LuoandFeng1} J. Luo and K. Feng, ``Cyclic codes and sequences from generalized Coulter-Matthews function,'' \emph{IEEE Trans. Inform. Theory},  \textbf{54}(12), 5345--5353, 2008.

\bibitem{LuoandFeng2}J. Luo and K. Feng, ``On the weight distributions of two classes of cyclic codes,'' \emph{IEEE Trans. Inform. Theory}, \textbf{54}(12), 5332--5344, 2008.

\bibitem{LuoHellseth2011} J. Luo  and T. Helleseth, "Constant Composition Codes as
Subcodes of Cyclic Codes", \emph{IEEE Trans. Inf. Theory},  \textbf{57}(11),  2011.

\bibitem{yuliuamc} L. Yu and H. Liu, "A class of p-ary cyclic codes
and their weight enumerators",  \emph{Adv.  Math.  Commun.}, \textbf{10}(2),   437--457, 2016.
\bibitem{yu-liu2016} L. Yu and H. Liu, ``The weight distribution of a family of p-ary cyclic codes'', \emph{Des. Codes Cryptogr}.  \textbf{78}, 731--745, 2016.

\bibitem{Zeng2012} X. Zeng,  J. Shan and  L. Hu, ``A triple-error-correcting cyclic code from the Gold and Kasami-Welch APN
power functions,'' \emph{Finite Fields Appl.} \textbf{18}(1), 70-92(2012).
\bibitem{ZengHujiang2010} X. Zeng,  L. Hu,  W. Jiang,  Q. Yue and  X. Cao, ``The weight distribution of a class of $p$-ary cyclic codes'', \emph{Finite
Fields Appl.} \textbf{16}(1), 56--73, 2010.
\bibitem{zhengwanghzeng2015} D. Zheng, X. Wang,
X. Zeng and L. Hu, ``The weight distribution of a family of $p$-ary cyclic codes,'' \emph{Des. Codes Cryptogr}. \textbf{75}(2), 1--13, 2015.


\bibitem{zhengwangyuliu2014} D. Zheng, X. Wang, L. Yu and H. Liu, "The weight distributions of several classes of p-ary cyclic codes", \emph{Discrete Mathematics}, \textbf{338},  1264--1276, 2015.
\bibitem{zhouding2014}Z. Zhou and C. Ding, ``A class of three-weight cyclic codes'', \emph{Finite Fields Appl.} \textbf{25}, 79--93, 2014.
\bibitem{ZhouLiFanH2016} Z. Zhou , N. Li, C. Fan and T. Helleseth, "Linear codes with two or three weights from
quadratic bent functions", \emph{Des Codes Cryptogr},   81(2), 283--295, 2016.
\end{thebibliography}
\end{document}